\documentclass[10pt, doublecolumn]{IEEEtran}
\usepackage{epsfig,latexsym}
\usepackage{float}
\usepackage{indentfirst}
\usepackage{amsmath}
\usepackage{amssymb}
\usepackage{times}
\usepackage{subfigure}
\usepackage{psfrag}
\usepackage{cite}
\usepackage{lastpage}
\usepackage{fancyhdr}
\usepackage{color}
 \usepackage{amsthm}
\usepackage{bigints}
\newtheorem{Lemma}{Lemma}
\newtheorem{lemma}[Lemma]{$\mathbf{Lemma}$}

\begin{document}
\title{ {  Relay Selection for Cooperative NOMA }}

\author{ Zhiguo Ding, \IEEEmembership{Senior Member, IEEE}, Huaiyu Dai, \IEEEmembership{Senior Member, IEEE},  and  H. Vincent Poor, \IEEEmembership{Fellow, IEEE}\thanks{
The authors   are with the Department of
Electrical Engineering, Princeton University, Princeton, NJ 08544,
USA.   Z. Ding is also with the School of
Computing and Communications, Lancaster
University, LA1 4WA, UK. H. Dai is  with the Department of Electrical and Computer
Engineering, North Carolina State University, Raleigh, NC 27695 USA.
}\vspace{-1.8em}} \maketitle
\begin{abstract}
 This letter studies the impact of relay selection (RS) on the performance of   cooperative non-orthogonal multiple access (NOMA). In particular, a two-stage  RS strategy is proposed, and analytical results are developed to demonstrate that this two-stage    strategy can achieve the minimal outage probability among all possible RS schemes, and realize the maximal diversity gain. The provided simulation results  show that cooperative NOMA with  this two-stage RS scheme outperforms that with the conventional max-min approach, and can also yield a significant performance gain over   orthogonal multiple access.
\end{abstract}\vspace{-1.3em}
\section{Introduction}
Non-orthogonal multiple access (NOMA) has been recognized as a promising enabling technology to improve the spectral efficiency of the fifth generation (5G) mobile network, and has been recently included into the fourth generation (4G) long term evolution (LTE) system \cite{6692652,Nomading,Krikidisnoma}. The application of cooperative transmission  to NOMA is important since   spatial degrees of freedom can be still harvested even if nodes are equipped with a single antenna.

A few different forms of cooperative NOMA have been proposed in the literature. The work in \cite{Zhiguo_conoma} relied on the cooperation among NOMA users, i.e., users with strong channel conditions act as relays. A dedicated relay has been  used in \cite{7230246} to improve the transmission reliability for a user with poor channel conditions. Similarly, a dedicated  relay has been used in \cite{7219393} to serve multiple users equipped with multiple antennas. Wireless power transfer has been applied to cooperative NOMA in \cite{yuanweijsac}, as an incentive for user cooperation.

This letter is to consider a downlink communication scenario with one base station,  two users and multiple relays. The impact of relay selection on the performance of cooperative NOMA will be studied, where we will focus on   two types of relay selection  criteria. The first one is based on conventional max-min relay selection \cite{4801494}. The second one is carried out in a two-stage strategy, where the first stage is to ensure one user's targeted data rate realized, and the second is to maximize the other user's rate opportunistically. We obtain a closed form expression for the outage probability achieved by the two-stage relay selection strategy,  which shows that this two-stage scheme can realize the maximal diversity gain. Furthermore, analytical results are developed to demonstrate that the two-stage strategy is also outage-optimal, i.e., it achieves the optimal outage probability among all possible relay selection schemes. On the other hand,  the max-min relay selection criterion can achieve the same performance as the two-stage one, i.e., realizing the minimal outage probability, for a special case with symmetrical setups, but it  suffers a loss of the outage probability in general. \vspace{-1em}
\section{System Model}
Consider a downlink   scenario with one base station (BS), two users, and $N$ relays. Each node is equipped with a single antenna. Assume that there is no direct link between the BS and the users, and the BS-relay and relay-user channels experience identically and independent Rayleigh fading. Unlike \cite{6692652,Nomading,Krikidisnoma}, users are not ordered by their channel conditions, but categorized by their quality of service (QoS) requirements.
Particularly,   assume that user $1$ is to be served  for small packet transmission, i.e., quickly connected with a low data rate, and user $2$ is to be served opportunistically \cite{Zhiguo_iot}. For example, user $1$ can be a healthcare sensor which is to send safety critical information containing in a few bytes, such as  heart rates or blood pressure. On the other hand, user $2$ is to download a movie, or perform background tasks.

During the first time slot, the BS will transmit the superimposed mixture, $(\alpha_1s_1+\alpha_2s_2)$, where $s_i$ denotes the signal to user $i$, $\alpha_i$ denotes the power allocation coefficient. Note that $\alpha_1^2+\alpha_2^2=1$ and $\alpha_1\geq \alpha_2$ in order to meet user $1$'s QoS requirements \cite{Zhiguo_iot}. Therefore, relay $n$, $1\leq n\leq N$, observes
\begin{align}
y^r_{n} = h_n (\alpha_1s_1+\alpha_2s_2)+w^r_{n},
\end{align}
where $h_n$ denotes the channel gain between the BS and relay $n$, and $w^r_{n}$ denotes the additive Gaussian noise.

The conditions for a relay to decode the two signals, $s_1$ and $s_2$, are given by{\small
\begin{align}\label{condition relay}
\log \left(1+\frac{|h_n|^2\alpha_1^2}{|h_n|^2\alpha_2^2+\frac{1}{\rho}}\right)\geq R_{1}, ~~ \log(1+\rho|h_n|^2\alpha_2^2)>R_2,
\end{align}
}\noindent where $\rho$ denotes the transmit signal-to-noise ratio (SNR) and $R_i$ is the targeted data rate for user $i$.

During the second time slot, assume that relay $n$ can decode the two signals and is selected to send  $(\alpha_1s_1+\alpha_2s_2)$.  Therefore, user $i$ receives the following:
\begin{align}
y^d_{n,i} = g_{n,i} (\alpha_1s_1+\alpha_2s_2)+w^d_{n,i},~i\in \{1, 2\},
\end{align}
where $g_{n,i}$ denotes the channel gain between relay $n$ and user $i$ and $w^d_{n,i}$ denotes the additive Gaussian noise. User $1$ decodes its message with the   signal-to-interference-plus-noise ratio (SINR), $
\frac{|g_{n,1}|^2\alpha_1^2}{|g_{n,1}|^2\alpha_2^2+\frac{1}{\rho}}$,
and user $2$ decodes its own message with the SNR, $\rho \alpha_2^2 |g_{n,2}|^2$, provided that $\log \left(1+\frac{|g_{n,2}|^2\alpha_1^2}{|g_{n,2}|^2\alpha_2^2+\frac{1}{\rho}}\right)\geq R_{1}$. Note that fixed power allocation is used in this paper. Optimizing the power allocation coefficients  
 and also  using  different power allocation policies for differen time slots can further improve the performance of cooperative NOMA, which is out of the scope of this paper.
\vspace{-1em}

\subsection*{Relay Selection Strategies}
\subsubsection{Max-min relay selection}
The criterion for this type of relay selection can be obtained   as follows \cite{4801494}:
\begin{align}\label{n*}
\max & \left\{ \min \{ |h_n|^2, |g_{n,1}|^2, |g_{n,2}|^2 \}, n \in\mathcal{S}_r\right\},
\end{align}
which is to select a relay with  the strongest $\min \{ |h_n|^2, |g_{n,1}|^2, |g_{n,2}|^2 \}$.

\subsubsection{Two-stage relay selection} The aim of this relay selection strategy is to realize  two purposes simultaneously. One is to ensure user $1$'s targeted data rate is realized, and the other is to serve user $2$ with a rate as large as possible. Specifically, this two-stage  user selection strategy can be described in the following.   The first stage is to build the following   subset of the relays by focusing on user $1$'s targeted data rate:
\begin{align}\nonumber
\mathcal{S}_r = &\left\{n: 1\leq n\leq N, \frac{1}{2}\log \left(1+\frac{|h_n|^2\alpha_1^2}{|h_n|^2\alpha_2^2+\frac{1}{\rho}}\right)\geq R_{1}, \right.\\ \nonumber &\frac{1}{2}\log \left(1+\frac{|g_{n,1}|^2\alpha_1^2}{|g_{n,1}|^2\alpha_2^2+\frac{1}{\rho}}\right)\geq R_{1}\\   &\left. \frac{1}{2}\log \left(1+\frac{|g_{n,2}|^2\alpha_1^2}{|g_{n,2}|^2\alpha_2^2+\frac{1}{\rho}}\right)\geq R_{1} \right\}.\label{stateg 1}
\end{align}

Denote the size of $\mathcal{S}_r$ by $|\mathcal{S}_r|$. Among the relays in $\mathcal{S}_r$, the second stage is to select a relay which can maximize the rate for user $2$, i.e.,
\begin{align}\label{n*}
n^*=\underset{n}{\arg}~\max & \left\{ \min \{\log(1+\rho|h_n|^2\alpha_2^2),\right. \\ \nonumber &\left.~~~~~~~\log(1+\rho|g_{n,2}|^2\alpha_2^2)\}, n \in\mathcal{S}_r\right\}.
\end{align}
\section{Performance Analysis}
In this section, we will characterize the outage probability achieved by the two-stage relay selection scheme. Note that the overall outage event can be categorized as follows:
\begin{align}
\mathcal{O}= \mathcal{O}_{1}\bigcup \mathcal{O}_{2} ,
\end{align}
where $\mathcal{O}_{1}$ denotes the event that relay $n^*$ cannot decode   $s_1$, or either  of the two users cannot decode $s_1$ successfully,  and $\mathcal{O}_{2}$ denotes the event that $s_2$ cannot be decoded correctly either by relay $n^*$, or by user $2$, while $s_1$ can be decoded correctly by the three nodes.

Therefore, the outage probability can be written as follows:
\begin{align}\label{eq1}
\mathrm{P}(\mathcal{O})= \mathrm{P}(\mathcal{O}_{1})+\mathrm{P}(\mathcal{O}_{2}) .
\end{align}
The term $ \mathrm{P}(\mathcal{O}_{1}) $ can be calculated as follows:
\begin{align}\label{individual}
  \mathrm{P}(\mathcal{O}_{1}) = & \mathrm{P} (|\mathcal{S}_r|=0)\\ \nonumber =& \prod_{n=1}^{N}\left[1-\mathrm{P}\left(|h_n|^2>\xi_1\right) \right.\\ \nonumber &\times \left. \mathrm{P}\left(|g_{n,1}|^2>\xi_1\right) \mathrm{P}\left(|g_{n,2}|^2>\xi_1\right)\right],
\end{align}
where $\xi_1=\frac{\frac{\epsilon_1}{\rho}}{\alpha_1^2 - \epsilon_1\alpha^2_2}$  and $\epsilon_1=2^{2R_1}-1$. It is assumed that $\alpha_1^2 > \epsilon_1\alpha^2_2$,  otherwise the outage probability is always one, a phenomenon  also observed in \cite{Nomading}.  By using the fact that all channels are assumed to be Rayleigh fading, we can have
\begin{align}\label{eq2}
  \mathrm{P}(\mathcal{O}_{1}) =  \prod_{n=1}^{N}\left[1-e^{-3\xi_1 }\right] .
\end{align}

The term $ \mathrm{P}(\mathcal{O}_{2}) $ can be calculated as follows:
\begin{align}\label{O2}
\mathrm{P}(\mathcal{O}_{2}) = \mathrm{P}(E_{1},  |\mathcal{S}_r|>0)+\mathrm{P}( {E_2},\bar{E}_1,|\mathcal{S}_r|>0),
\end{align}
where $E_1$ denotes the event that relay $n^*$ cannot decode $s_2$, $\bar{E}_1$ denotes the complementary event  of $E_1$,  and $E_2$ denotes that user $2$ cannot decode $s_2$. The first term in the above equation can be expressed as follows:
\begin{align}
 &\mathrm{P}(E_{1}, |\mathcal{S}_r|>0)\\ \nonumber &= \mathrm{P}\left(\log(1+\rho|h_{n^*}|^2\alpha_2^2)<2R_2,|\mathcal{S}_r|>0\right) ,
 \end{align}
where $\xi_2=\frac{2^{2R_2}-1}{\rho\alpha_2^2}$.

 The second term in \eqref{O2} can be expressed as follows:
\begin{align} \nonumber
\mathrm{P}( {E_2},\bar{E}_1,|\mathcal{S}_r|>0)   =& \mathrm{P}\left(\log(1+\rho|g_{n^*,2}|^2\alpha_2^2)<2R_2, \right.\\ \nonumber &\left. \log(1+\rho|h_{n^*}|^2\alpha_2^2)>2R_2,|\mathcal{S}_r|>0\right) .
\end{align}
Therefore, the probability $ \mathrm{P}(\mathcal{O}_{2}) $ can be calculated as follows:
\begin{align}
\mathrm{P}(\mathcal{O}_{2}) =& \mathrm{P}\left(\log(1+\rho|h_{n^*}|^2\alpha_2^2)<2R_2,|\mathcal{S}_r|>0\right) \\ \nonumber &+
 \mathrm{P}\left(\log(1+\rho|g_{n^*,2}|^2\alpha_2^2)<2R_2, \right.\\ \nonumber &\left. \log(1+\rho|h_{n^*}|^2\alpha_2^2)>2R_2,|\mathcal{S}_r|>0\right).
\end{align}
Assuming $|\mathcal{S}_r|>0$,   define
\begin{align}
x_n= \min \{\log(1+\rho|h_n|^2\alpha_2^2),\log(1+\rho|g_{n,2}|^2\alpha_2^2)\},
\end{align}
and
\begin{align}
x_{n^*}= \max \{x_i, \forall i \in\mathcal{S}_r\} .
\end{align}
The probability $ \mathrm{P}(\mathcal{O}_{2}) $ can now be expressed as follows:
\begin{align}
\mathrm{P}(\mathcal{O}_{2}) =& \mathrm{P}\left(\min\left\{\log(1+\rho|h_{n^*}|^2\alpha_2^2), \right.\right.  \\ \nonumber &\left.\left.  \log(1+\rho|g_{n^*,2}|^2\alpha_2^2)<2R_2\right\} ,|\mathcal{S}_r|>0\right)\\ \nonumber =&\mathrm{P}\left(x_{n^*}<2R_2 ,|\mathcal{S}_r|>0\right).
\end{align}
The above probability can further expressed as follows:
\begin{align}
\mathrm{P}(\mathcal{O}_{2}) =& \sum^{N}_{l=1}\mathrm{P}\left(x_{n^*}<2R_2 ,|\mathcal{S}_r|=l\right)
\\ \nonumber =&\sum^{N}_{l=1}\mathrm{P}\left(x_{n^*}<2R_2 ||\mathcal{S}_r|=l\right)\mathrm{P}\left( |\mathcal{S}_r|=l\right).
\end{align}

For a relay randomly selected from $\mathcal{S}_r$, denoted by relay $n$, the cumulative distribution function (CDF) of $x_n$ can be founded as follows:
\begin{align}\nonumber
F (x) =& \mathrm{P}\left(\min \{ |h_n|^2,|g_{n,2}|^2\}<\left.\frac{2^x-1}{\rho\alpha_2^2}\right| n\in\mathcal{S}_r, |\mathcal{S}_r|\neq 0\right)\\ \nonumber =&
  \mathrm{P}\left( |h_n|^2 >|g_{n,2}|^2 ,|g_{n,2}|^2 <\left.\frac{2^x-1}{\rho\alpha_2^2}\right|\right.\\ \nonumber &\left. |h_n|^2>\xi_1, |g_{n,2}|^2>\xi_1 \right)\\ \nonumber &+
  \mathrm{P}\left( |h_n|^2 <|g_{n,2}|^2 ,|h_{n}|^2 <\left.\frac{2^x-1}{\rho\alpha_2^2}\right|\right.\\ \nonumber &\left. |h_n|^2>\xi_1, |g_{n,2}|^2>\xi_1 \right).
\end{align}
Define the two probabilities at the right hand side of the above equation by $Q_1$ and $Q_2$, respectively. The probability $Q_1$ can be expressed as follows:
\begin{align}\nonumber
Q_1 =& \frac{
  \mathrm{P}\left( |h_n|^2 >|g_{n,2}|^2,|g_{n,2}|^2 < y,  |h_n|^2>\xi_1, |g_{n,2}|^2>\xi_1 \right)}{\mathrm{P}\left( |h_n|^2>\xi_1, |g_{n,2}|^2>\xi_1 \right)}\\ \nonumber =&\frac{
  \mathrm{P}\left( |h_n|^2>\max\left\{\xi_1, |g_{n,2}|^2 \right\},\xi_1<|g_{n,2}|^2< y  \right)}{\mathrm{P}\left( |h_n|^2>\xi_1, |g_{n,2}|^2>\xi_1 \right)}
\end{align}
where $y=\frac{2^x-1}{\rho\alpha_2^2}$. The constraint on $y$,  $y\geq \xi_1$, will be explained later. By using the Rayleigh  assumption, we have
\begin{align}\nonumber
Q_1 =&   e^{2\xi_1}\int_{\xi_1}^{ y} e^{-\max\left\{\xi_1, z \right\}-z}  dz \\ \nonumber =&  \frac{1}{2}e^{2\xi_1}\left( e^{-2\xi_1}-e^{- 2y}\right).
\end{align}
$Q_2$ can be obtained similarly, and therefore,  the CDF can be expressed  as follows:
\begin{align}\nonumber
F (x) =& e^{2\xi_1}\left( e^{-2\xi_1}-e^{- 2\frac{(2^x-1)}{\rho\alpha_2^2}}\right).
\end{align}
It is important to point out the following:
\begin{align}
x&=\log(1+\rho\min \{|h_n|^2,|g_{n,2}|^2\} \alpha_2^2) \\ \nonumber &\geq \log\left(1+\rho \xi_1\alpha_2^2\right),
\end{align}
which is due to the fact that both $|h_n|^2$ and $|g_{n,2}|^2$ should be larger than $\xi_1$, since relay $n$ is in  $ \mathcal{S}_r$. With this constraint, one can easily verify that
\begin{align}
F(\log\left(1+\rho \xi_1\alpha_2^2\right))=0,
\end{align}
and $F(\infty)=1$. 
With this CDF, the   probability for $\mathcal{O}_{2}$ can be calculated  as follows:
\begin{align}\label{eq3}
\mathrm{P}(\mathcal{O}_{2}) =& \sum^{N}_{l=1}\mathrm{P}\left(x_{n^*}<2R_2 ||\mathcal{S}_r|=l\right)\mathrm{P}\left( |\mathcal{S}_r|=l\right)
\\ \nonumber =& \sum^{N}_{l=1}\left(F(2R_2)\right)^l\mathrm{P}\left( |\mathcal{S}_r|=l\right).
\end{align}
On the other hand, the probability to have $l$ relays in $\mathcal{S}_r$ can be calculated as follows:
\begin{align}
\mathrm{P}\left( |\mathcal{S}_r|=l\right)   =& {N \choose l} \prod_{n=1}^{N-l}\left[1-\mathrm{P}\left(|h_{\pi(n)}|^2>\xi_1\right) \right.\\ \nonumber &\times \left. \mathrm{P}\left(|g_{{\pi(n)},1}|^2>\xi_1\right) \mathrm{P}\left(|g_{{\pi(n)},2}|^2>\xi_1\right)\right]\\ \nonumber &\times  \prod_{n=N-l+1}^N\left[\mathrm{P}\left(|h_{\pi(n)}|^2>\xi_1\right) \right.\\ \nonumber &\times \left. \mathrm{P}\left(|g_{{\pi(n)},1}|^2>\xi_1\right) \mathrm{P}\left(|g_{{\pi(n)},2}|^2>\xi_1\right)\right],
\end{align}
where ${\pi(\cdot)}$ denotes a random permutation of the relays. Following steps similar to those used to obtain \eqref{individual}, the above probability can be obtained as follows:
\begin{align}\label{eq4)}
\mathrm{P}\left( |\mathcal{S}_r|=l\right)   =& {N \choose l}  \left[1-e^{-3\xi_1 }\right]^{N-l} e^{-3l\xi_1}.
\end{align}
By combing \eqref{eq1}, \eqref{eq2}, \eqref{eq3}, and \eqref{eq4)}, and also applying some algebraic manipulations,   the overall outage probability can be obtained in the following lemma.
\begin{lemma}\label{lemma1}
The overall outage probability achieved by the two-stage relay selection scheme can be expressed  as follows:
\begin{align}\label{leam eq1}
\mathrm{P}(\mathcal{O})=& \sum^{N}_{l=0}{N \choose l}\left(F(2R_2)\right)^l e^{-3l\xi_1}  \left[1-e^{-3\xi_1}\right]^{N-l},
\end{align}
if $\alpha_1^2 > \epsilon_1\alpha^2_2$, otherwise $\mathrm{P}(\mathcal{O})=1$.
\end{lemma}

{\it Remark 1:}
At high SNR, $\rho$ approaches infinity,   $\xi_1$  approaches zero, which means that  the function  $F(2R_2)$ can be approximated as follows:
\begin{align}\nonumber
F (2R_2) =& e^{2\xi_1}\left( e^{-2\xi_1}-e^{- 2\frac{(2^{2R_2}-1)}{\rho\alpha_2^2}}\right)\\ \nonumber \approx&\left(  2\frac{(2^{2R_2}-1)}{\rho\alpha_2^2}-2\xi_1\right)=\frac{\gamma}{\rho},
\end{align}
where $\gamma= 2\frac{(2^{2R_2}-1)}{ \alpha_2^2}-2\frac{ \epsilon_1 }{\alpha_1^2 - \epsilon_1\alpha^2_2}$.
By using the above approximation, the overall outage probability can be approximated  as follows:
\begin{align}
\mathrm{P}(\mathcal{O})\approx&\sum^{N}_{l=0}{N \choose l}\frac{\gamma^l}{\rho^l} e^{-3l\xi_1 }  \left[1-e^{-3\xi_1}\right]^{N-l}\\ \nonumber
\approx&\frac{1}{\rho^N}\sum^{N}_{l=0}{N \choose l}\gamma^l   \left[\frac{ 3\epsilon_1 }{\alpha_1^2 - \epsilon_1\alpha^2_2} \right]^{N-l}.
\end{align}
 Therefore, the two-stage RS scheme can realize a diversity gain of $N$, which is the maximal diversity gain given the existence of the $N$ relays.

{\it Remark 2:} The optimality of the two-stage relay selection scheme is shown in the following lemma.
\begin{lemma}\label{lemma2}
For the addressed cooperative NOMA scenario, the two-stage relay selection scheme minimizes the overall outage probability.
\end{lemma}
\begin{proof} The lemma can be proved by contradiction. If there exists a better strategy  achieving a lower outage probability, an event that the use of relay $n^*$ causes outage, but no outage occurs with  the relay selected by the new strategy, denoted by $\bar{n}^*$, $\bar{n}^*\neq n^*$, should happen.   Recall that for any relay selection scheme, the outage event can be categorized  as follows:
\begin{align}\label{prof1}
\mathcal{O}= \mathcal{O}_{1}\bigcup \mathcal{O}_{2}.
\end{align}
 We only focus on the cases with $|\mathcal{S}_r|\neq 0$, otherwise outage always occurs, no matter which relay  is used.  When $|\mathcal{S}_r|>0$, one can conclude that relay $\bar{n}^*$ must be in $\mathcal{S}_r$, i.e., $\bar{n}^*\in \mathcal{S}_r$, otherwise outage occurs for sure by using relay $\bar{n}^*$.
 According to \eqref{stateg 1}, relay $n^*$ will not cause outage event $\mathcal{O}_1$ as well, if $|\mathcal{S}_r|\neq 0$.   Now by using the criterion in \eqref{n*} and the definition of $\mathcal{O}_2$, one can conclude that it is not possible that relay $n^*$ causes $\mathcal{O}_2$ but relay $\bar{n}^*$ does not, since relay $n^*$ is the optimal solution to avoid $\mathcal{O}_2$. The lemma is proved.
 \end{proof}
 {\it Remark 3:} Simulation results show that the two-stage  relay selection scheme outperforms the max-min scheme. However, for a special case with symmetrical setups, e.g.,  $\xi_1=\xi_2$, we can show that the two schemes achieve the same performance. The overall outage probability can be expressed as follows:
\begin{align}
\mathrm{P}(\mathcal{O}) &= \mathrm{P} \left(|h_n|^2<\xi_1\right)  + \mathrm{P} \left(   |h_n|^2 <\xi_2,|h_n|^2>\xi_1\right)\\ \nonumber &+  \mathrm{P} \left( |g_{n,1}|^2<\xi_1, |h_n|^2 >\xi_2,|h_n|^2>\xi_1\right)
\\ \nonumber &+  \mathrm{P} \left( |g_{n,2}|^2<\xi_1, |g_{n,1}|^2>\xi_1, |h_n|^2 >\xi_2,|h_n|^2>\xi_1\right)
\\ \nonumber &+  \mathrm{P} \left( |g_{n,2}|^2<\xi_2, |g_{n,2}|^2>\xi_1, |g_{n,1}|^2>\xi_1, |h_n|^2 >\xi_2,\right. \\ \nonumber &\left.~~~~~~|h_n|^2>\xi_1\right).
\end{align}
When $\xi_1=\xi_2$, we can have
\begin{align}
\mathrm{P}_o &=   \mathrm{P} \left(   |h_n|^2 < \xi_1\right)\\ \nonumber &+  \mathrm{P} \left( |g_{n,1}|^2<\xi_1, |h_n|^2 >\xi_1\right)
\\ \nonumber &+  \mathrm{P} \left( |g_{n,2}|^2<\xi_1, |g_{n,1}|^2>\xi_1, |h_n|^2 > \xi_1\right).
\end{align}
Note that the following equality holds
\begin{align}
  &   \mathrm{P} \left(   |h_n|^2 < \xi_1\right) +  \mathrm{P} \left( |g_{n,1}|^2<\xi_1, |h_n|^2 >\xi_1\right)
\\ \nonumber &=   \mathrm{P} \left(\min\{ |g_{n,1}|^2, |h_n|^2\}<\xi_1\right).
\end{align}
By using this equality, the outage probability achieved by the max-min approach is given by
\begin{align}\nonumber
\mathrm{P}_o &=     \mathrm{P} \left(\min\{|g_{n,2}|^2, |g_{n,1}|^2, |h_n|^2\}<\xi_1, \forall n\in\{1, \cdots, N\}\right)\\ \nonumber &=\left[1-\mathrm{P} \left(\min\{|g_{\pi(1),2}|^2, |g_{\pi(1),1}|^2, |h_{\pi(1)}|^2\}>\xi_1\right)\right]^N \\   &= [1-e^{-3\xi_1}]^N,
\end{align}
which is exactly the same as Lemma \ref{lemma1} by applying $\xi_1=\xi_2$.
\vspace{-0.8em}
\begin{figure}[!htbp]\centering\vspace{-0.5em}
    \epsfig{file=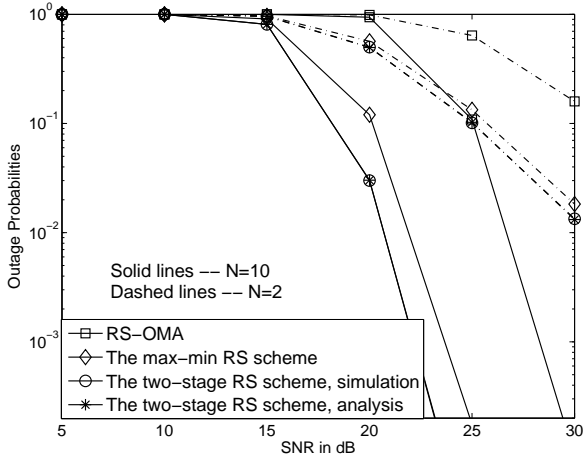, width=0.45\textwidth, clip=}\vspace{-0.5em}
\caption{ Comparison between cooperative OMA and   NOMA with different relay selection (RS) strategies. $R_1=0.5$ bit per channel use (BPCU), $R_2=2$ BPCU, and $\alpha_1=\frac{1}{4}$.   The analytical results are based on Lemma \ref{lemma1}.  }\label{fig1}\vspace{-1em}
\end{figure}

\vspace{-0.8em}
\section{Numerical Studies}
In this section, the performance of   cooperative NOMA   with the   two relay selection strategies is evaluated by using computer simulations. In Fig. \ref{fig1}, the performance of cooperative NOMA is compared with that of orthogonal multiple access (OMA). For OMA, $4$ time slots are needed, and the max-min criterion is used for relay selection. As can be observed from Fig. \ref{fig1}, cooperative NOMA can efficiently reduce the outage probability, and hence the use of cooperative NOMA can offer a significant performance gain over OMA in terms of reception reliability. The reason for this performance gain is that the use of NOMA can ensure that two users are served simultaneously, whereas two times of bandwidth resources, such as time slots, are needed for OMA to serve the two users.
\begin{figure}[!htbp]\centering\vspace{-1em}
    \epsfig{file=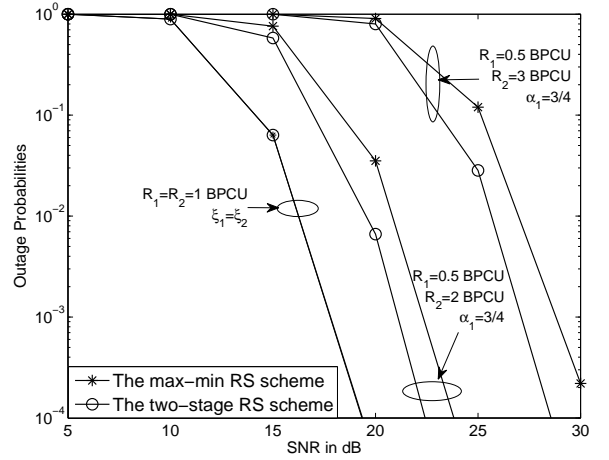, width=0.45\textwidth, clip=}\vspace{-0.5em}
\caption{ The outage probabilities achieved by the max-min relay selection scheme and the two-stage one. $N=10$. }\label{fig2}\vspace{-1em}
\end{figure}

The performance difference between the max-min relay selection scheme and the two-stage one is also illustrated in Figs. \ref{fig1} and \ref{fig2}.  When $\xi_1\neq \xi_2$,  the two-stage relay selection scheme outperforms the max-min scheme, and this observation is consistent with Lemma \ref{lemma2} which shows that  the two-stage relay selection scheme achieves the minimal outage probability. When the number of the relays is small, the performance gap between the two relay selection schemes is small, and the use of more relays can increase this gap. One can also observe that the simulation results perfectly match the analytical results developed in Lemma \ref{lemma1}, which demonstrates the accuracy of the developed analytical results. Furthermore, when $\xi_1=\xi_2$, the two relay selection schemes achieve the same performance, as discussed in Remark 3 in the previous section.\vspace{-0.5em}
\section{Conclusions}
In this paper, we have studied the impact of relay selection on cooperative NOMA. Particularly two types of relay selection have been proposed and studied, where a closed form expression for the outage probability achieved by the two-stage scheme has been obtained. The developed analytical results have demonstrated that the two-stage scheme can achieve not only the optimal diversity gain, but also the minimal outage probability. Compared to the two-stage scheme, the max-min relay selection criterion results in a loss of the outage probability, except in a special case with symmetrical setups.

\vspace{-0.5em}
 \bibliographystyle{IEEEtran}
\bibliography{IEEEfull,trasfer}

\begin{thebibliography}{1}
\providecommand{\url}[1]{#1}
\csname url@samestyle\endcsname
\providecommand{\newblock}{\relax}
\providecommand{\bibinfo}[2]{#2}
\providecommand{\BIBentrySTDinterwordspacing}{\spaceskip=0pt\relax}
\providecommand{\BIBentryALTinterwordstretchfactor}{4}
\providecommand{\BIBentryALTinterwordspacing}{\spaceskip=\fontdimen2\font plus
\BIBentryALTinterwordstretchfactor\fontdimen3\font minus
  \fontdimen4\font\relax}
\providecommand{\BIBforeignlanguage}[2]{{%
\expandafter\ifx\csname l@#1\endcsname\relax
\typeout{** WARNING: IEEEtran.bst: No hyphenation pattern has been}%
\typeout{** loaded for the language `#1'. Using the pattern for}%
\typeout{** the default language instead.}%
\else
\language=\csname l@#1\endcsname
\fi
#2}}
\providecommand{\BIBdecl}{\relax}
\BIBdecl

\bibitem{6692652}
Y.~Saito, Y.~Kishiyama, A.~Benjebbour, T.~Nakamura, A.~Li, and K.~Higuchi,
  ``Non-orthogonal multiple access ({NOMA}) for cellular future radio access,''
  in \emph{Proc. IEEE Veh. Tech. Conference}, Dresden, Germany, Jun. 2013.

\bibitem{Nomading}
Z.~Ding, Z.~Yang, P.~Fan, and H.~V. Poor, ``On the performance of
  non-orthogonal multiple access in {5G} systems with randomly deployed
  users,'' \emph{IEEE Signal Process. Lett.}, vol.~21, no.~12, pp. 1501--1505,
  Dec. 2014.

\bibitem{Krikidisnoma}
S.~Timotheou and I.~Krikidis, ``Fairness for non-orthogonal multiple access in
  {5G} systems,'' \emph{IEEE Signal Process. Lett.}, vol.~22, no.~10, pp.
  1647--1651, Oct. 2015.

\bibitem{Zhiguo_conoma}
Z.~Ding, M.~Peng, and H.~V. Poor, ``Cooperative non-orthogonal multiple access
  in {5G} systems,'' \emph{IEEE Commun. Lett.}, vol.~19, no.~8, pp. 1462--1465,
  Aug. 2015.

\bibitem{7230246}
J.-B. Kim and I.-H. Lee, ``Non-orthogonal multiple access in coordinated direct
  and relay transmission,'' \emph{IEEE Commun. Lett.}, vol.~19, no.~11, pp.
  2037--2040, Nov. 2015.

\bibitem{7219393}
J.~Men and J.~Ge, ``Non-orthogonal multiple access for multiple-antenna
  relaying networks,'' \emph{IEEE Commun. Lett.}, vol.~19, no.~10, pp.
  1686--1689, Oct. 2015.

\bibitem{yuanweijsac}
Y.~Liu, Z.~Ding, M.~Elkashlan, and H.~V. Poor, ``Cooperative non-orthogonal
  multiple access with simultaneous wireless information and power transfer,''
  \emph{IEEE J. Sel. Areas Commun.}, 2016, (to appear).

\bibitem{4801494}
Y.~Jing and H.~Jafarkhani, ``Single and multiple relay selection schemes and
  their achievable diversity orders,'' \emph{IEEE Transactions on Wireless
  Communications}, vol.~8, no.~3, pp. 1414--1423, Mar. 2009.

\bibitem{Zhiguo_iot}
Z.~Ding, L.~Dai, and H.~V. Poor, ``{MIMO-NOMA} design for small packet
  transmission in the internet of things,'' \emph{IEEE Access}, (to appear in
  2016).

\end{thebibliography}

  \end{document}